\documentclass[letterpaper]{article} 
\usepackage{aaai2026}  
\usepackage{times}  
\usepackage{helvet}  
\usepackage{courier}  
\usepackage[hyphens]{url}  
\usepackage{graphicx} 
\usepackage{natbib}  
\usepackage{caption} 
\usepackage{algorithm}
\usepackage{algorithmic}

\usepackage{amsmath} 
\usepackage{amsthm}
\newtheorem{theorem}{Theorem}
\usepackage{amsfonts}       
\usepackage{booktabs}
\usepackage{subcaption}  

\usepackage{newfloat}
\usepackage{listings}
\DeclareCaptionStyle{ruled}{labelfont=normalfont,labelsep=colon,strut=off} 
\floatstyle{ruled}
\newfloat{listing}{tb}{lst}{}
\floatname{listing}{Listing}
\title{MARPO: A Reflective Policy Optimization for Multi-Agent Reinforcement Learning}
\author{
    Cuiling Wu\textsuperscript{\rm 1,\rm 2,}\equalcontrib,
    Yaozhong Gan\textsuperscript{\rm 2,}\equalcontrib,
   Junliang Xing\textsuperscript{\rm 2,}\footnotemark[2],
    Ying Fu\textsuperscript{\rm 1,}\thanks{Co-corresponding author.}
}
\affiliations{
    \textsuperscript{\rm 1} School of Computer Science and Technology, Beijing Institute of Technology \\
    \textsuperscript{\rm 2} QiYuan Lab\\

    wucuiling@bit.edu.cn,
    yzgancn@163.com,
    xingjunliang@qiyuanlab.com,
    fuying@bit.edu.cn

}

\usepackage{bibentry}

\begin{document}

\maketitle

\begin{abstract}
We propose \textbf{M}ulti-\textbf{A}gent \textbf{R}eflective \textbf{P}olicy \textbf{O}ptimization (\textbf{MARPO}) to alleviate the issue of sample inefficiency in multi-agent reinforcement learning. MARPO consists of two key components: a \textbf{reflection mechanism} that leverages subsequent trajectories to enhance sample efficiency, and an \textbf{asymmetric clipping mechanism} that is derived from the KL divergence and dynamically adjusts the clipping range to improve training stability. We evaluate MARPO in classic multi-agent environments, where it consistently outperforms other methods. 
\end{abstract}

\section{Introduction}

Multi-Agent Reinforcement Learning (MARL) has become a vital tool for complex decision-making tasks across areas such as autonomous driving, robotics, and cooperative games. Despite its potential, MARL faces challenges, particularly in achieving sample efficiency in high-dimensional environments. Unlike supervised learning, which uses static datasets, reinforcement learning (RL) relies on agent-environment interactions to collect samples—a costly process in both time and resources. 

Early RL advancements, such as Trust Region Policy Optimization (TRPO)~\citep{schulman2015trust} and Proximal Policy Optimization (PPO)~\citep{schulman2017proximal}, have improved stability; however, they still face scalability and adaptability issues in real-world environments. Many MARL algorithms, including IPPO~\citep{de2020independent}, MAPPO~\citep{yu2022surprising}, QMIX~\citep{rashid2020monotonic}, HATRPO, and HAPPO~\citep{kuba2021trust}, rely on frequent interactions with the environment and optimize policies based on individual state-action pairs, often underutilizing trajectory-level information that can enhance policy stability.

Recent research explores methods to improve sample efficiency. Trajectory Reuse Optimization, such as R2D2~\citep{kapturowski2018recurrent}, accelerates convergence by prioritizing experience replay and reusing long sequences. Model-based methods, such as MuZero~\citep{silver2017mastering}, reduce real-world interactions by simulating environment dynamics internally. Additionally, innovations in Sampling Policy and Reward mechanisms, such as Prioritized Experience Replay (PER)~\citep{schaul2015prioritized} and Inverse Reinforcement Learning (IRL)~\citep{ng2000algorithms}, prioritize high-value experiences to improve learning efficiency. Frameworks such as Efficient Episodic Memory Utilization (EMU)~\citep{na2024efficient} and Episodic Multi-agent Reinforcement Learning with Curiosity-driven exploration (EMC)~\citep{zheng2021episodic} further enhance exploration and policy convergence.



While advancements in MARL have been made, challenges in sample efficiency, scalability, and trajectory-level utilization remain. To address these, we introduce \textbf{Multi-Agent Reflective Policy Optimization (MARPO)}, a framework that leverages trajectory feedback to improve policy optimization efficiency. Unlike methods that rely on auxiliary components such as value functions or human feedback~\citep{deng2024improving, haarnoja2018soft,he2025decode}, MARPO directly optimizes policies using trajectory signals, thereby enhancing learning efficiency and decision-making. We propose an asymmetric KL-divergence-based clipping mechanism with a dynamic adjustment strategy for the clipping range, further improving training stability and flexibility. We validate MARPO on the StarCraft II Multi-Agent Challenge (SMAC)~\citep{samvelyan19smac}, including its more complex SMAC-Hard variants~\citep{deng2024smach}, SMACv2~\citep{ellis2023smacv2}, and Google Research Football (GRF)~\citep{kurach2020google}, demonstrating superior performance over MAPPO in terms of sample efficiency and policy optimization.

\textbf{Our contributions can be summarized as follows:}
\begin{itemize}
    \item We propose MARPO, the first framework to integrate reflection mechanisms into multi-agent policy optimization, improving sample efficiency by leveraging trajectory feedback.
    \item We derive a KL-based asymmetric clipping mechanism that enables more accurate and flexible policy updates.
    \item We introduce a dynamic adaptation strategy for the clipping range, which enhances the exploration capability compared to fixed-boundary approaches such as PPO.
\end{itemize}

\section{Related Work}

\subsection{KL-based Asymmetric Clipping Mechanism}

Trust Region Policy Optimization (TRPO)~\citep{schulman2015trust} ensures stable policy updates by using a trust-region constraint that prevents substantial updates that could destabilize the learning process. It optimizes the policy using a second-order approximation, specifically the Fisher information matrix, to account for the impact of policy updates on the value function~\citep{prokopenko2011relating}. The goal is to maximize expected return while keeping the KL divergence between the old and new policies within a predefined threshold, ensuring steady progress and stability during training.

Proximal Policy Optimization (PPO)~\citep{schulman2017proximal}, an extension of TRPO, simplifies computation by using first-order approximations and a clipping mechanism to limit deviations between the new and old policies. This approach improves sample efficiency and computational feasibility, making PPO scalable to large-scale problems. However, the fixed clipping range in PPO is a limitation in dynamic MARL environments, where continuous interactions between agents require more adaptive mechanisms~\citep{massaoudi2025adaptive,liu2024jointppo}. This has led to research on dynamic clipping mechanisms, such as Trust Region-Guided Proximal Policy Optimization (TRGPPO)~\citep{wang2019trust}, which adjusts the clipping range based on KL divergence. This approach enables larger updates in regions with lower KL divergence, promoting exploration, while restricting updates in regions with higher KL divergence to maintain stability. Although it improves sample efficiency and exploration, it relies on careful tuning of the KL threshold.

\subsection{Reflective Policy Optimization}

Reflective Policy Optimization (RPO)~\citep{pmlr-v235-gan24b} extends on-policy reinforcement learning methods such as TRPO and PPO by integrating future-trajectory feedback into policy updates. Although PPO is widely adopted for its simplicity and empirical robustness, it remains limited in sample efficiency~\citep{xiong2023sample}. By leveraging information from the whole trajectory, RPO stabilizes learning and promotes monotonic improvement, enabling agents to adjust actions with a more introspective, trajectory-aware update rule. Proposed initially for single-agent RL, RPO naturally extends to MARL, where agents jointly reflect on both current and future trajectories to improve coordination and accelerate convergence. Similar ideas have also been explored for multi-agent systems under communication delays~\citep{qin2024virtual}, further demonstrating the stability advantages of future-aware updates.

In MARL, MAPPO~\citep{yu2022surprising} builds on PPO under the CTDE paradigm: agents train with global information while acting from local observations. This mitigates non-stationarity~\citep{hernandez2017survey,canese2021multi} via centralized value functions and stabilized importance-sampling updates. However, its reliance on current state–action pairs restricts the exploitation of future-trajectory structure, limiting stability and scalability in complex environments~\citep{hernandez2017survey, canese2021multi, zhang2025unaligned,jin2025comprehensive}.

Our work unifies these lines by retaining PPO-style practicality, adding KL-guided adaptive clipping to avoid brittle fixed thresholds, and injecting future-trajectory information into the update. This yields steadier improvement under multi-agent non-stationarity while remaining computationally lightweight. Beyond our on-policy MARL focus, proximal principles also appear in federated offline RL—e.g., proximal evaluation under distributed data constraints~\citep{10748547,10410601}—and in continuous-control RL via DDPG-family proximal variants~\citep{10606213}, pointing to scalable pathways for deployment.

\section{Preliminaries}

\subsection{Cooperative MARL Problem Formulation}

In the context of MARL, the challenge of partially observable environments is modeled using the Decentralized Partially Observable Markov Decision Process (Dec-POMDP) ~\citep{oliehoek2016concise}. In the Dec-POMDP framework for cooperative multi-agent tasks, the environment is modeled by a tuple \( G = \langle  \mathcal{N}, S, \mathcal{A}, P, O, r, \gamma \rangle \), where \( S \) is the finite state space, each agent \( i \in \mathcal{N} \equiv \{1, \dots, n\} \) chooses an action \( a_i \in A \),
 \( A \) is the individual action space, and \( P(s' | s, A) \) defines the state transition dynamics, 
\(r:S \times \mathcal{A} \rightarrow \mathbb{R} \text{ is the reward function, and } \gamma \in [0, 1] \text{ is the discount factor.}\)
 Agents use a policy \( \pi_{\theta}(a_i | o_i) \) parameterized by \( \theta \) to produce an action \( a_i \) from the local observation \( o_i \), and jointly optimize the discounted accumulated reward
\(J(\theta) = \mathbb{E}_{A^t, s^t} \left[ \sum_t \gamma^t R(s^t, A^t) \right]\)
where \( A^t = (a^t_1, \dots, a^t_n) \) is the joint action at time step \( t \).

\subsection{Reflection Mechanism}

 Reflective Policy Optimization (RPO) extends the clipped surrogate objective to leverage interactions over multiple steps, aiming to address the substantial data requirements per update common to TRPO and PPO, which often lead to sample inefficiency. 
 The surrogate objective function of RPO is defined as:
\begin{align}
L(\pi,\pi_{\text{old}}) = L^{\text{clip}}_0(\pi,\pi_{\text{old}}) + \alpha L^{\text{clip}}_1(\pi, \pi_{\text{old}}),
\end{align}
where 
\begin{equation}
    \begin{aligned}
    L^{\text{clip}}_0(\pi,\pi_{\text{old}}) = & \, \mathbb{E}_{(s, a)} \left[ \min \left( \rho(a|s) A^{\pi_{\text{old}}}(s, a), \right. \right. \\
    & \, \left. \left. \text{clip}(\rho(a|s), 1-\epsilon, 1+\epsilon) A^{\pi_{\text{old}}}(s, a) \right) \right],
    \end{aligned}
\end{equation}
\begin{equation}
    \begin{aligned}
    L^{\text{clip}}_1(\pi, \pi_{\text{old}}) = & \, \mathbb{E}_{(s, a, s', a')} \left[ \min \left( \rho(a|s) \rho(a'|s') A^{\pi_{\text{old}}}(s', a'), \right. \right. \\
    & \, \left. \left. C(\rho, \rho') A^{\pi_{\text{old}}}(s', a') \right) \right],
    \end{aligned}
    \label{eq:ref_ma}
\end{equation}
and
 \( \rho(a|s) = \frac{\pi(a|s)}{\pi_{\text{old}}(a|s)} \), \( C(\rho, \rho') = \text{clip}(\rho(a|s), 1-\epsilon, 1+\epsilon)  \cdot \text{clip}(\rho(a' | s'), 1-\varepsilon_1, 1+\varepsilon_1) \), $\epsilon$, $\epsilon_1$ and \(\alpha\) are the hyperparameter. 
The essence of the reflection mechanism is captured in Eqn.~(\ref{eq:ref_ma}), and the theoretical foundation of RPO ensures monotonic improvements in policy performance.
 As a result, RPO enhances sample efficiency while maintaining stability, addressing the computational and performance challenges in policy optimization.

\subsection{Challenges of Current Policy Gradient Approaches in MARL}
Extending Policy Gradient methods to MARL presents significant challenges. A simple approach is to use a shared parameter set for all agents and aggregate their trajectories for policy optimization. This strategy, implemented in MAPPO, optimizes the policy parameter \(\theta\) by maximizing an objective function that combines the benefits of centralized training with decentralized execution, ensuring stability and efficiency in multi-agent environments. The objective function is defined as:
\begin{equation}
    \begin{aligned}
    L(\theta) = &\; \frac{1}{n} \textstyle \sum_{i=1}^n \mathbb{E}_{(o_i^k, a_i^k)} \left[
    \min \left( \rho_{\theta, i}^k A_i^k,\,
    \text{clip}(\rho_{\theta, i}^k, 1-\epsilon, \right. \right. \\
    & \left. \left. 1+\epsilon) A_i^k \right)
    \right] + \sigma \cdot \frac{1}{n} \textstyle \sum_{i=1}^n \mathbb{E}_{o_i^k} \left[
    S\left[\pi_{\theta}^k(o_i^k)\right]
    \right],
    \end{aligned}
    \label{eq:final_loss}
\end{equation}
where \( \rho_{\theta, i}^k = \pi_\theta^k(a_i^k \mid o_i^k) / \pi_{\theta_{\text{old}}}^k(a_i^k \mid o_i^k) \) is the policy ratio for agent \(i\), and \( A_i^k \) is the advantage function computed using the Generalized Advantage Estimation (GAE) method~\citep{schulman2015high}. The clip term ensures that policy updates remain within a trust region, thereby mitigating large, destabilizing changes. \( S \) represents the policy entropy, and \( \sigma \) is the entropy coefficient hyper parameter.


\section{Method}
In light of the challenges associated with sample inefficiency, we introduce the Multi-Agent Reflective Policy Optimization (MARPO) framework. Unlike MAPPO, which represents a straightforward extension of PPO to multi-agent environments, MARPO constitutes a natural and principled evolution of RPO. It not only enhances sample efficiency by incorporating a joint reflection mechanism but also improves exploration capabilities via an adaptive clipping mechanism. The details of the algorithm are outlined in Algorithm \ref{alg:MARPO}.

\begin{algorithm}[t]
\caption{MARPO}
\label{alg:MARPO}
\textbf{Input}: Initial policy parameters $\theta$; hyperparameter $\alpha$ \\
\textbf{Parameter}: Number of iterations $n$, number of epochs $K$ \\
\textbf{Output}: Updated parameters $\theta$
\begin{algorithmic}[1]
    \FOR{$\text{iteration} = 1$ to $n$}
        \STATE Collect trajectories for all agents and store them in dataset $D$
        \FOR{$\text{epoch} = 1$ to $K$}
            \STATE Sample mini-batch $b$ from dataset $D$
            \STATE Compute clipping bounds $x_1$ and $x_2$ by Eqn.~(\ref{eq:clip_bound})
            \STATE Compute policy loss $L_{\theta}(\pi, \pi_{\text{old}})$ by Eqn.~(\ref{eq:policy_loss})
            \STATE Update policy parameters: $\theta \gets \theta - \alpha \nabla_\theta L_{\theta}$
        \ENDFOR
    \ENDFOR
\end{algorithmic}
\end{algorithm}

\subsection{Description of the MARPO Method}
We propose a multi-agent framework that combines reflection and clipping to enhance sample efficiency. For clarity of exposition, we omit explicit time-step subscripts when not ambiguous. The overall objective function is defined as:
\begin{align}
    L(\pi, \pi_{\text{old}}) 
    = L_{0}^{\text{clip}}(\pi, \pi_{\text{old}}) + \alpha L_{1}^{\text{clip}}(\pi, \pi_{\text{old}}),
    \label{eq:policy_loss}
\end{align}
where 

\begin{equation}
    \begin{aligned}
    L_{0}^{\text{clip}}(\pi, \pi_{\text{old}}) 
    =\, & \frac{1}{n} \sum^{n}_{i=1} \mathbb{E}_{(o_i^k, a_i^k)} \Big[ \min \Big(
    \rho_{i}^k(a_i^k | o_i^k) \cdot A^{\pi_{\text{old}}^k}, \\ 
    & \text{clip}(\rho_{i}^k(a_i^k | o_i^k), x_1, x_2) \cdot A^{\pi_{\text{old}}^k}
    \Big) \Big],
    \end{aligned}
    \label{eq:multi_agent_clip_loss}
\end{equation}
and 
\begin{equation}
\begin{aligned}
L_{1}^{\text{clip}}(\pi,\pi_{\text{old}})
&= \frac{1}{n}\sum_{i=1}^{n}\mathbb{E}_{(o_i^k,a_i^k,o_i^{k+1},a_i^{k+1})}
   \Big[\min\Big(\\
&\quad \rho_i^k(a_i^k\mid o_i^k)\,
      \rho_i^{k+1}(a_i^{k+1}\mid o_i^{k+1})\, A^{\pi_{\text{old}}^{k}},\\
&\quad c(\rho_i^k,\rho_i^{k+1})\, A^{\pi_{\text{old}}^{k+1}}
\Big)\Big],
\end{aligned}
\end{equation}
where \(\rho_{i}^k\) denotes the probability ratio between the new policy \(\pi\) and the old policy \(\pi_{\text{old}}\) for agent \(i\) at time step \(k\), given the local observation \(o_i^k\) and the action \(a_i^k\) taken. Similarly,
\(
\rho_i^{k+1}(a_i^{k+1} | o_i^{k+1}) = \pi^{k+1}(a_i^{k+1} | o_i^{k+1}) / \pi_{\text{old}}^{k+1}(a_i^{k+1} | o_i^{k+1})
\)
represents the ratio for the action \(a_i^{k+1}\) at the next local state \( o_i^{k+1}\). Additionally, \(A^{\pi_{\text{old}}^k}\) is the advantage function calculated from the old policy by GAE, which helps estimate the value of each state-action pair in a multi-agent context. The hyperparameter \(\alpha\) controls the clipping bounds and balances the trade-off between different components of the objective function.

The clipping mechanism ensures that policy updates remain within a reasonable range, thereby avoiding large updates that could destabilize the training process. It is defined as:
\begin{equation}
\begin{aligned}
c(\rho_i^k, \rho_i^{k+1}) &= \text{clip}(\rho_i^k(a_i^k | o_i^k), x_1, x_2) \\
&\cdot \text{clip}(\rho_i^{k+1}(a_i^{k+1} | o_i^{k+1}), x_1', x_2'),
\end{aligned}
\label{eq:clip_formula}
\end{equation}
where \( x_1, x_2 \) and \( x_1', x_2' \) are clipping bounds that regulate the magnitude of policy updates, by incorporating future state-action pairs, agents can account for the long-term impact of their actions, thereby enhancing global coordination and overall policy quality. This integration leverages future state-action information to support more informed and stable policy updates. Our approach enhances coordination, accelerates training, and stabilizes learning, effectively addressing key challenges commonly faced by traditional MARL methods.

\subsection{Dynamic Adaptive Asymmetric Clipping Mechanism}

In the previous section, we introduced a multi-agent reflective mechanism in which the policy optimization objective considers not only the reward at the current time step but also the potential influence of future actions. However, under the reflection mechanism, the two importance sampling ratios $\rho^k$ and $\rho^{k+1}$ co-occur in the loss, making the objective more sensitive to policy drift. To tackle this issue, we introduce a Dynamic Asymmetric Clipping Mechanism, which leverages approximate estimates of the KL divergence to dynamically adjust the bounds $x_1$ and $x_2$, thereby enhancing the stability and constraint of policy updates.

\noindent \textbf{Theoretical Analysis and Enhancement of Clipping Mechanisms Using KL Divergence.} 
The standard clipping mechanism in PPO can be viewed as a heuristic approximation to the TV distance, designed to constrain the magnitude of policy updates. Specifically, the TV distance between the old and new policies can be expressed as:
\begin{equation}
D_{\text{TV}}(\pi_{\text{old}} || \pi_{\text{new}}) 
= \frac{1}{2} \sum_{a} |\pi_{\text{old}} - \pi_{\text{new}}| 
= \frac{1}{2} \, \mathbb{E}_{\pi_{\text{old}}} \left| 1 - \frac{\pi_{\text{new}}}{\pi_{\text{old}}} \right|.
\end{equation}

To ensure stable policy updates, PPO constrains the importance sampling ratio \(\left|1-\frac{\pi_{\text{new}}}{\pi_{\text{old}}}\right|\) within a fixed range, effectively bounding policy deviation through a surrogate loss that approximates total variation (TV) distance. While TV distance controls the magnitude of policy shifts, it only reflects the scalar distance between distributions and overlooks the geometric structure of the policy space. 

In contrast, KL divergence provides an asymmetric and differentiable measure that more accurately captures directional changes in policies, making it better suited for gradient-based optimization, particularly in high-dimensional or complex multi-agent reinforcement learning. Formally, the KL divergence between the old and new policies is defined as:
\begin{equation}
D_{\text{KL}}(\pi_{\text{old}} || \pi_{\text{new}}) 
= \sum_{a} \pi_{\text{old}} \cdot \log \frac{\pi_{\text{old}}}{\pi_{\text{new}}} 
= \mathbb{E}_{\pi_{\text{old}}} \left( \log \frac{\pi_{\text{old}}}{\pi_{\text{new}}} \right).
\label{eq:kl_divergence}
\end{equation}

Although Eqn.~(\ref{eq:kl_divergence}) offers a rigorous definition, directly constraining \( \log \frac{\pi_{\text{old}}}{\pi_{\text{new}}} \) can lead to incorrect behavior, as it does not provide a proper unbiased estimator and may even yield negative values, violating the non-negativity of KL divergence. To address this, we first establish the following identity, which forms the foundation for an unbiased KL estimator:
\begin{align}
 &E_{\pi_{\text{old}}} \left( -\log \frac{\pi_{\text{new}}}{\pi_{\text{old}}} + \frac{\pi_{\text{new}}}{\pi_{\text{old}}} - 1 \right) \nonumber \\
 =& E_{\pi_{\text{old}}} \left( \log \frac{\pi_{\text{old}}}{\pi_{\text{new}}} + \frac{\pi_{\text{new}}}{\pi_{\text{old}}} - 1 \right) \nonumber \\
 =& D_{\text{KL}}(\pi_{\text{old}} || \pi_{\text{new}})  + \sum \pi_{\text{new}} - \sum \pi_{\text{old}} \nonumber \\
 =&  D_{\text{KL}}(\pi_{\text{old}} || \pi_{\text{new}}). 
 \label{eq:unbaised_ev}
\end{align}
where the last equality follows from the fact that both policies are normalized distributions, i.e. \(\sum_a \pi_{\text{new}} = \sum_a \pi_{\text{old}} = 1.\)  

This motivates the definition of \(f(x) = x - 1 - \log x\),
where \(x = \frac{\pi_{\text{new}}}{\pi_{\text{old}}}\). Crucially, \(f(x)\) satisfies \(E_{\pi_{\text{old}}}[f(x)] = D_{\text{KL}}(\pi_{\text{old}} || \pi_{\text{new}})\), thereby providing an unbiased and theoretically sound surrogate for KL divergence that is inherently non-negative. This property allows \(f(x)\) to serve as a robust alternative to the heuristic TV-based clipping mechanism in PPO and forms the basis of our enhanced clipping strategy.

\begin{figure}[t]
    \centering
    \includegraphics[width=0.47\textwidth]{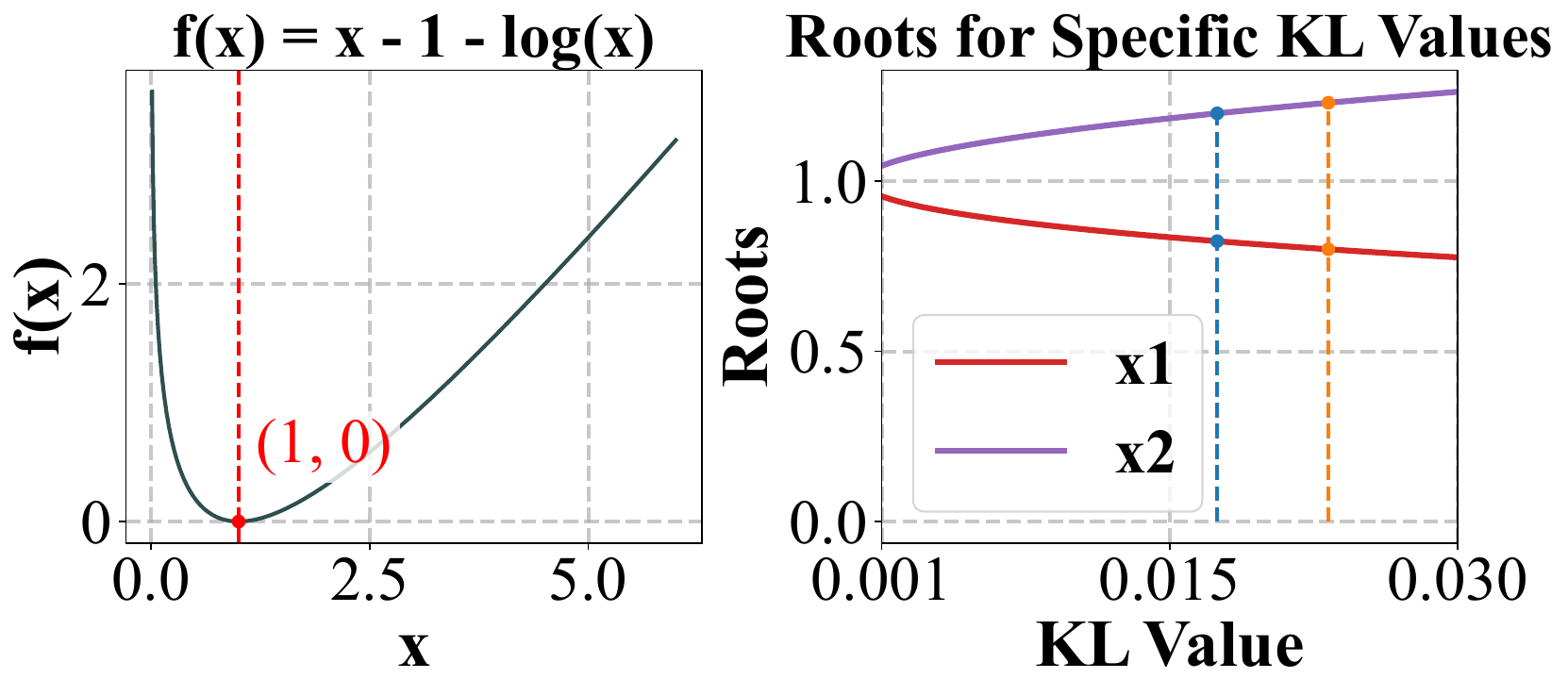}
    \caption{Root solving for dynamic clipping bounds based on KL targets.}
    \label{fig:kl_root}
\end{figure}
\begin{theorem}
We define the unbiased estimation function \( f(x) = x - 1 - \log(x) \), derived from Eqn.~(\ref{eq:unbaised_ev}), where \(x = \frac{\pi_{\text{new}}}{\pi_{\text{old}}}\). This function satisfies:
\begin{enumerate}
    \item \(f(x) \geq 0 \) for all \( x > 0 \) with equality if and only if \(x = 1\).
    \item \(f(x)\) is convex in \(x\), ensuring stable optimization when used for clipping.
\end{enumerate}
\end{theorem}

\begin{proof}
(1) \textbf{Non-negativity:}  
We first compute the derivative of \(f(x)\):
\[
f'(x) = 1 - \frac{1}{x}.
\]
For \(x > 1\), \(f'(x) > 0\), and for \(0 < x < 1\), \(f'(x) < 0\). Therefore, \(f(x)\) decreases on \((0,1]\) and increases on \([1, \infty)\), attaining its global minimum at \(x = 1\). Evaluating \(f(1)\):
\(f(1) = 1 - 1 - \log 1 = 0\). Thus, \(f(x) \geq 0\) for all \(x > 0\), with equality only at \(x=1.\)

(2) \textbf{Convexity:}  
We compute the second derivative:
\[
f''(x) = \frac{1}{x^2} > 0, \quad \forall x > 0,
\]
which implies that \(f(x)\) is convex. Convexity ensures that local clipping constraints defined by \(f(x)\) form a valid and stable bound in optimization.

\end{proof}

\noindent \textbf{Advantages of the unbiased estimator.}  
The proposed function \(f(x)\) exhibits several advantages over the standard TV-based clipping mechanism in PPO:
\begin{itemize}
    \item \textbf{Non-negativity and convexity:} \(f(x)\) avoids invalid negative divergence values and provides a convex structure for stable optimization.
    \item \textbf{Compatibility with KL-based policy optimization:} By aligning clipping directly with KL divergence, the method captures both the magnitude and directional information of policy updates, improving optimization in high-dimensional reinforcement learning settings.
\end{itemize}


\begin{figure*}[t]
    \centering
    \includegraphics[width=0.9\textwidth]{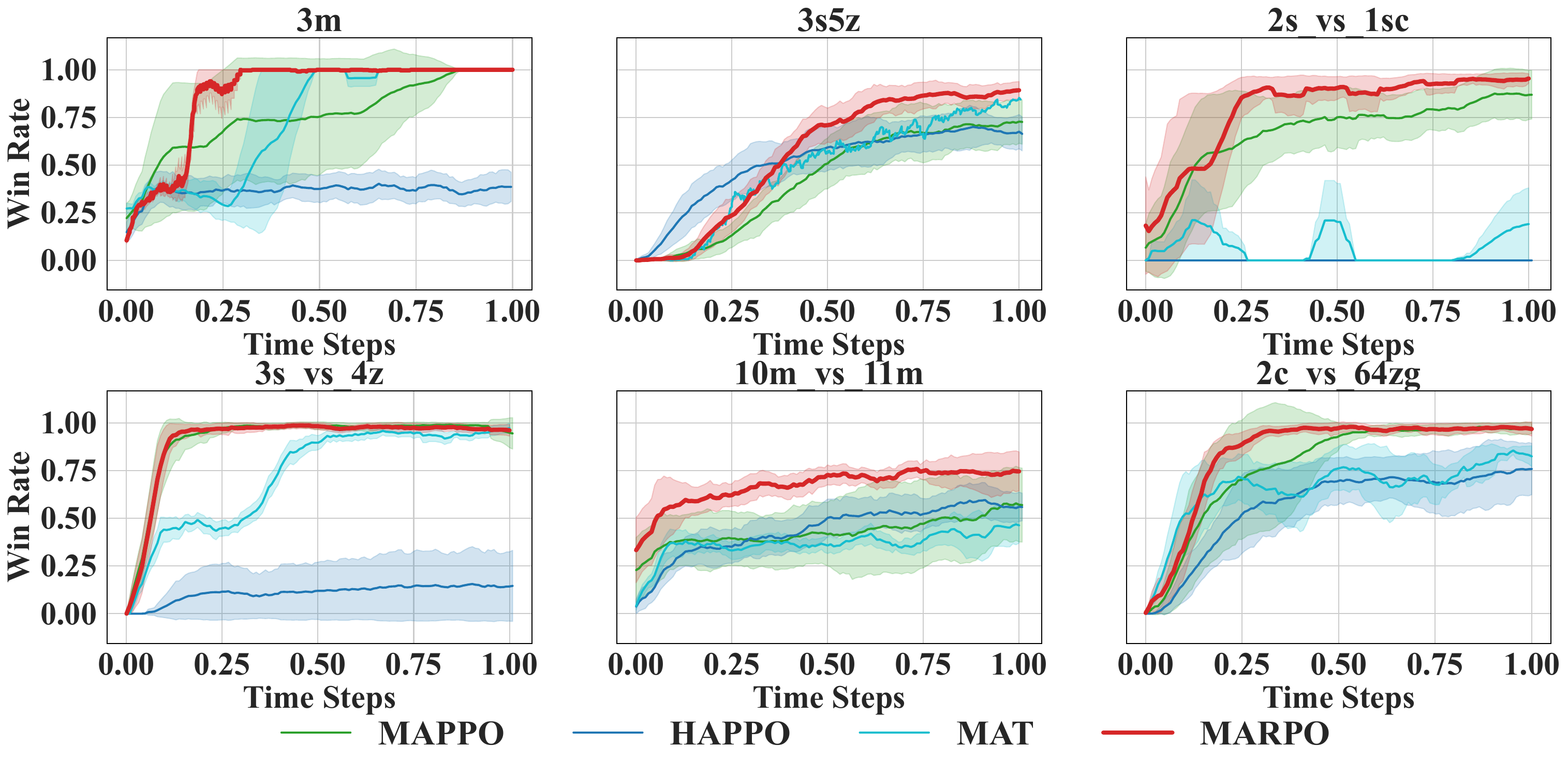}
\caption{Win rate curves for the main experiments on SMAC-Hard environments. Shaded regions represent the standard deviation across five random seeds. The X-axis denotes timesteps, and the Y-axis denotes win rate. This experiment primarily compares policy-based methods with baselines.}

    \label{fig:performance_curves}
\end{figure*}

\begin{figure*}[t]
    \centering
    \includegraphics[width=0.9\textwidth]{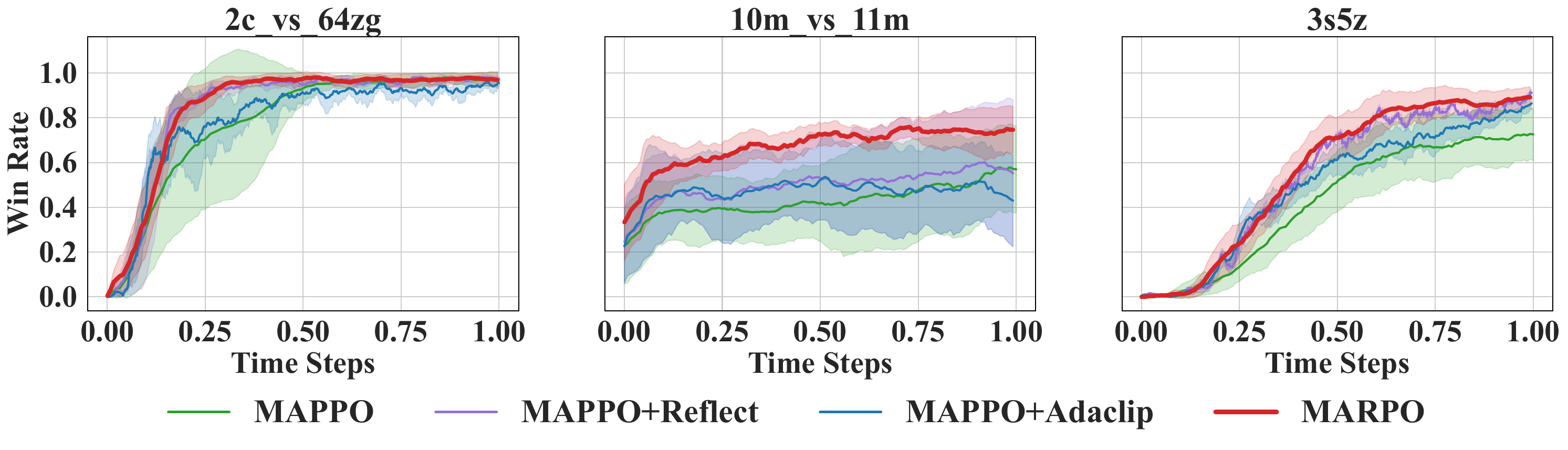}
    \caption{Win rate curves from the ablation study on SMAC-Hard environments, validating the impact of different modules. Shaded regions represent the standard deviation across five random seeds. }

    \label{fig:abation}
\end{figure*}

\noindent \textbf{Construction of Clipping Bounds.} 
To enhance the stability and adaptability of reflective policy optimization, we propose a novel clipping mechanism that dynamically determines asymmetric clipping bounds based on the inverse of a function \( f(x) \) related to the true KL divergence. Specifically, the true KL divergence between the current and the old policy is computed directly, and the function \( f(x) \) is used to find the inverse, yielding two roots that define the clipping interval. This target KL value is subsequently smoothed using an exponential moving average (EMA), forming the dynamic KL constraint. 
\begin{equation}
D_{\mathrm{KL_t}}^{\mathrm{target}} = \beta \cdot D_{\mathrm{KL_{t-1}}}^{\mathrm{target}} + (1 - \beta) \cdot D_{\mathrm{KL_t}}.
\label{eq:clip_bound}
\end{equation}

In contrast to fixed or heuristically scheduled KL bounds, our formulation induces a self-adjusting trust region that evolves automatically with the policy. The clipping thresholds are obtained by numerically solving for the roots of a KL-based approximation function, as introduced in Eqn.~(\ref{eq:clip_bound}), yielding an asymmetric range $[x_1, x_2]$ that adapts continuously to the current policy and observed KL statistics. In our approach, these dynamic thresholds are derived analytically from the target KL divergence and updated online, which is further smoothed via an EMA mechanism to ensure temporal consistency and reduce short-term variance. As a result, the method eliminates the need for manually chosen KL bounds or hand-crafted annealing schedules. It offers a principled, stable, and adaptive clipping mechanism---distinct from PPO’s symmetric and fixed-range clipping (see Fig.~\ref{fig:kl_root})---while maintaining robust performance across diverse environments and training regimes.


\begin{figure*}[t]
    \centering
    \includegraphics[width=0.95\textwidth]{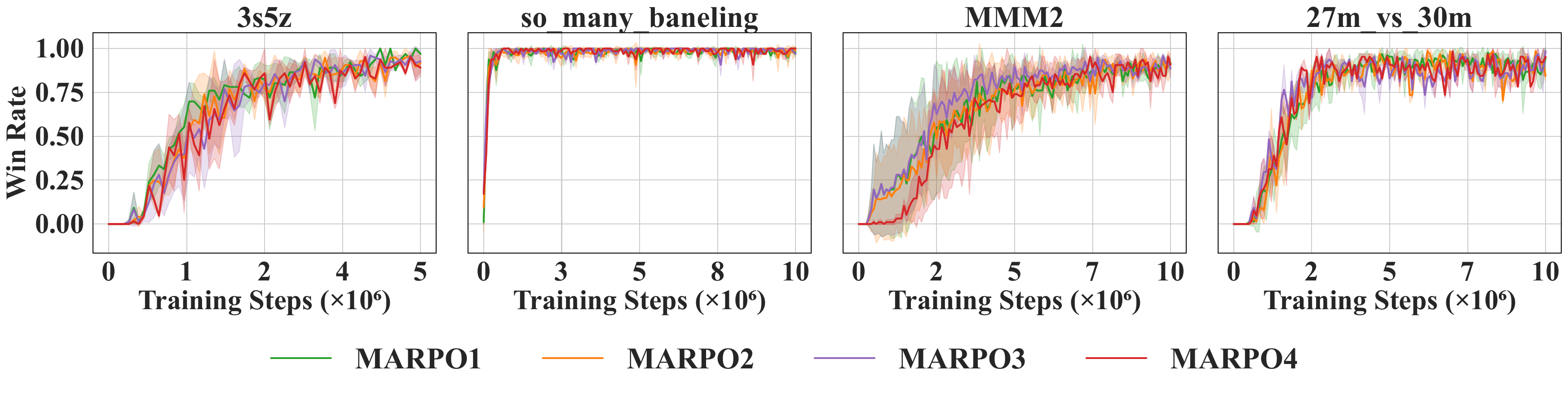}
    \caption{Hyperparameter Sensitivity Study: Impact of KL Bias and Sliding Average Update Rate (\(\beta\)) on Win Rate in SMAC Environments. Shaded regions represent the standard deviation across five random seeds. }
    \label{fig:kl_sensitivity}
\end{figure*}

\begin{figure*}[t]
    \centering
    \includegraphics[width=0.95\textwidth]{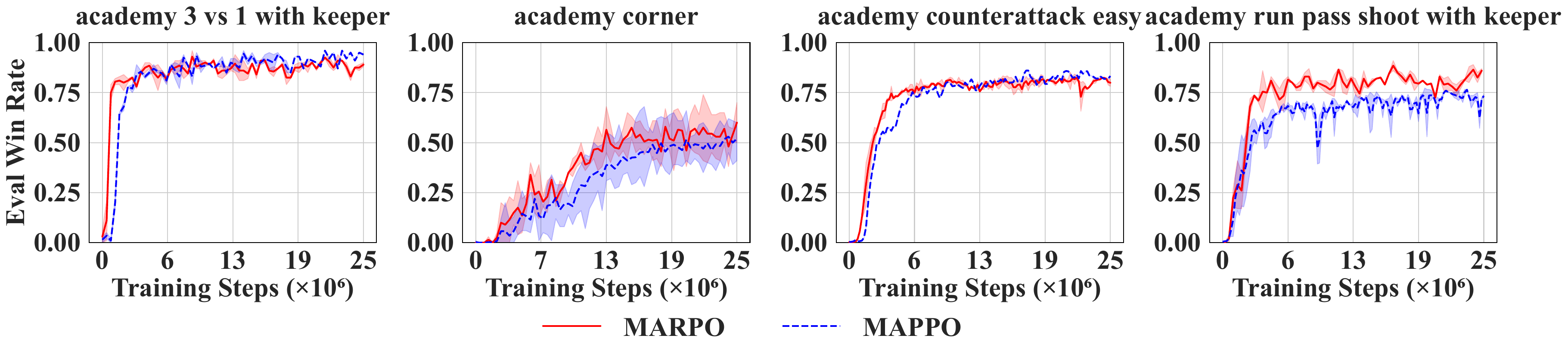}
    \caption{Win rate curves for the main experiments on GRF environments. Shaded regions represent the standard deviation across three random seeds. The X-axis denotes timesteps, and the Y-axis denotes win rate.}

    \label{fig:GRF_exp}
\end{figure*}

\begin{figure*}[t]
    \centering
    \includegraphics[width=0.95\textwidth]{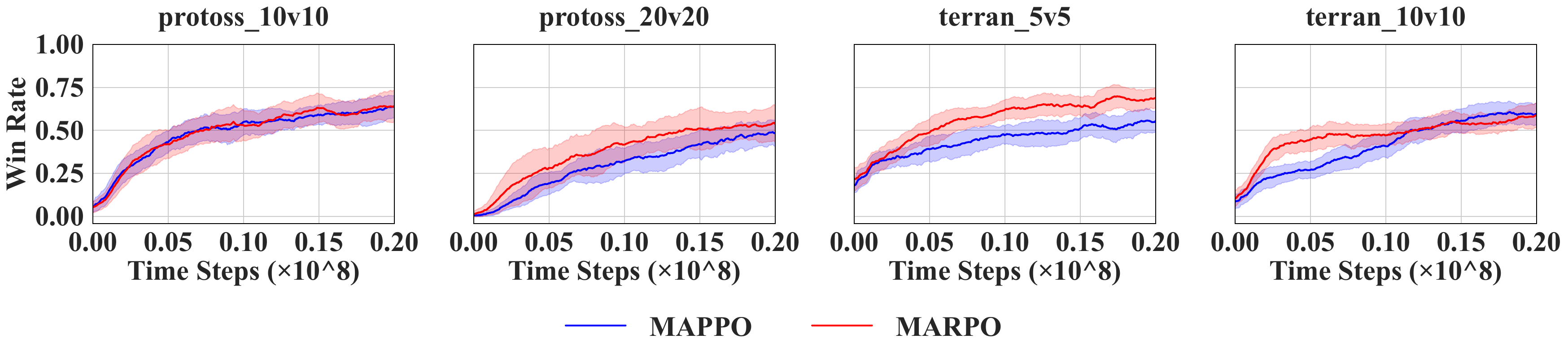}
    \caption{Win rate curves for the main experiments on SMACv2 environments. Shaded regions represent the standard deviation across three random seeds. The X-axis denotes timesteps, and the Y-axis denotes win rate.}

    \label{fig:SMACv2_exp}
\end{figure*}

\section{Experiments}

\begin{table*}[ht]
\centering
\setlength{\tabcolsep}{4pt}  
\begin{tabular}{lccccccc}
\toprule
Env Name & MAPPO  & HAPPO  & LDSA & QMIX  & QPLEX & MAT & MARPO (Ours) \\
\midrule
3m           & 99.2 ± 4.7 & 37.3 ± 8.7 & 99.6 ± 1.5 & 99.8 ± 0.8 & 5.7 ± 19.3 & 99.9 ± 0.6 & \textbf{100.0 ± 0.0} (\checkmark) \\
3s5z         & 71.0 ± 1.7 & 68.1 ± 1.1(\checkmark) & 12.3 ± 0.7 & 34.9 ± 0.8 & 26.4 ± 3.0 & 79.7 ± 0.5 & \textbf{87.2 ± 0.4}  \\
2s\_vs\_1sc   & 85.2 ± 13.5 & 0.0 ± 0.0 & 93.9 ± 9.5 & 61.4 ± 36.7 & 81.5 ± 15.5 & 10.3 ± 15.0 & \textbf{94.8 ± 4.2} (\checkmark) \\
3s\_vs\_4z    & \textbf{97.8 ± 7.5} & 14.4 ± 18.6 & 95.0 ± 6.5 & 75.8 ± 20.2 & 78.8 ± 22.8 & 94.4 ± 6.3 & 97.0 ± 3.2 (\checkmark) \\
10m\_vs\_11m  & 53.0 ± 4.7 & 57.6 ± 1.0 & 45.7 ± 2.6 & 65.2 ± 1.1 & 0.4 ± 0.0 & 43.2 ± 1.7 & \textbf{74.3 ± 1.2} (\checkmark) \\
2c\_vs\_64zg  & 97.3 ± 3.3 & 73.3 ± 16.9 & 87.0 ± 9.5 & 69.7 ± 32.3 & 35.1 ± 32.5 & 79.7 ± 12.3 & \textbf{97.4 ± 3.3} (\checkmark) \\
\bottomrule
\end{tabular}
\caption{Performance comparison across 6 SMAC-Hard environments. The values represent the average performance over the last 2 million steps. A check mark (\checkmark) in parentheses indicates the highest value for the first 2.5 million steps. The values are reported as mean ± standard deviation, with the variance represented as $10^{-2}$.}
\label{tab:performance_summary}
\end{table*}

\subsection{Experiment Setup and Evaluations}
\textbf{StarCraft II Multi-Agent Challenge.} SMAC is a widely used benchmark for cooperative MARL, featuring diverse combat scenarios based on the StarCraft II engine. Agents must collaborate to defeat built-in AI opponents. However, the original SMAC's limited opponent diversity often leads to overfitting and poor generalization. To address this, SMAC-Hard introduces mixed scripted opponents, randomized strategy switching, and a self-play interface, enhancing adversarial variability and robustness evaluation. With its open-source implementation publicly available, SMAC-Hard establishes a new benchmark for evaluating the robustness, adaptability, and strategic coverage of MARL algorithms in dynamic, partially observable environments. We conduct experiments on both SMAC and SMAC-Hard, covering scenarios of varying difficulty and asymmetry to assess the adaptability and generalization of our method.

\noindent\textbf{Main Experiment Setup.} Experiments are conducted on a subset of scenarios from the SMAC-Hard benchmark, selected to reflect a broad range of task difficulty and coordination complexity.   
We compare our method against several strong baselines on SMAC-Hard, including value-based methods (\textbf{QMIX}~\citep{rashid2020monotonic}, \textbf{LDSA}~\citep{yang2022ldsa}, \textbf{QPLEX}~\citep{wang2021qplex}) and policy-based methods (\textbf{MAPPO}~\citep{yu2022surprising}, \textbf{HAPPO}~\citep{kuba2021trust},  and  \textbf{MAT}~\citep{wen2022multiagent}), providing a comprehensive evaluation across diverse algorithmic paradigms. All agents are trained for 10 million environment steps. For evaluation, we report the average win rate over the final 2 million training steps. To ensure fair comparisons, all methods adopt identical neural network architectures and optimization settings. Additionally, we conduct supplementary experiments on the original SMAC benchmark to further evaluate the generalization ability of our approach.

\subsection{Main Results on SMAC-Hard}

Our method consistently outperforms existing baselines across diverse SMAC-Hard scenarios, achieving higher win rates and improved robustness (Table~\ref{tab:performance_summary}). This gain stems from better utilization of low-probability, high-impact samples during updates, mitigating issues in clipped objectives. As shown in Figure~\ref{fig:performance_curves}, our approach exhibits stable convergence and enhanced sample efficiency in challenging multi-agent tasks.

\subsection{Analysis of Algorithm Performance}
\noindent\textbf{Ablation Experiments.} We conduct ablation experiments to evaluate the contribution of each module within our framework. Additionally, we analyze the computational overhead to assess the efficiency and scalability of the proposed algorithm. As shown in Figure~\ref{fig:abation}, we conducted ablation experiments in most environments by removing each module individually. The results demonstrate that removing any module results in performance degradation, indicating that each component contributes to the overall effectiveness of the algorithm. This highlights that all modules play essential roles across different scenarios, and their combined use is crucial for achieving improved performance.

\noindent\textbf{Sensitivity Analysis of Hyperparameters.}  In our method, the KL divergence is computed as the true KL between the new and old policy distributions and averaged over trajectories. Instead of imposing an upper bound, we maintain a lower bound to prevent the KL from collapsing to zero. This value is tracked using an Exponential Moving Average (EMA) with two hyperparameters: the KL bias and the update rate~$\beta$ defined in Eqn.~(\ref{eq:clip_bound}). To examine the algorithm’s sensitivity to these hyperparameters, we test several representative settings and analyze their effects on stability, convergence, and sample efficiency.

Figure~\ref{fig:kl_sensitivity} summarizes the results across four configurations---MARPO1 (0.05, 0.05), MARPO2 (0.08, 0.08), MARPO3 (0.10, 0.08), and MARPO4 (0.10, 0.01)---where each pair denotes the KL bias and EMA update rate used in the \emph{second}, KL-guided clipping term, while the standard PPO-style ratio clipping threshold is kept fixed across all runs. The curves are almost indistinguishable across these configurations, indicating that once the KL-based clipping is enabled, MARPO is fairly insensitive to the precise choice of these two hyperparameters in our tested range. This suggests that MARPO does not rely on careful tuning of the KL-related clipping term for stable performance, and a more exhaustive joint exploration of both clipping mechanisms is left for future work.

\subsection{Experiments on Google Football and SMACv2}

We evaluate MARPO on two cooperative multi-agent benchmarks: Google Research Football (GRF) and SMACv2. GRF features fast-paced, partially observable gameplay that stresses real-time coordination among agents, while SMACv2 introduces stochastic unit behaviors, delayed rewards, and non-stationary opponents, making it substantially harder than the original SMAC and a good test of robustness.

On GRF, we focus on generalization and sample efficiency. As shown in Figure~\ref{fig:GRF_exp}, MARPO attains a clear early performance lead over MAPPO across multiple scenarios, indicating that the reflective mechanism helps extract useful trajectory information more quickly and accelerates policy adaptation. Throughout training, MARPO maintains stable performance that remains competitive with MAPPO, suggesting that the early gains do not come at the cost of instability. On SMACv2, we assess robustness under stochastic and non-stationary dynamics. Figure~\ref{fig:SMACv2_exp} shows that MARPO yields consistently higher win rates and more stable learning than MAPPO on a range of maps, and often converges faster despite the increased task difficulty and noisy credit assignment.

Taken together, the GRF and SMACv2 results show that MARPO improves early-phase sample efficiency while preserving strong robustness in challenging multi-agent environments, supporting its effectiveness as a general-purpose on-policy method for complex cooperative control in realistic, large-scale, and highly non-stationary settings with diverse tasks, agents, difficulty levels, and evaluation scenarios, where its reflective mechanism and adaptive clipping strategy enable stable learning dynamics and reliable coordination across heterogeneous benchmarks.

\section{Conclusion}


We propose Multi-Agent Reflective Policy Optimization (MARPO), a framework that enhances sample efficiency through a reflection mechanism and stabilizes training via a KL-guided asymmetric clipping strategy. Experiments on StarCraft II and Google Research Football demonstrate that MARPO achieves competitive performance across diverse settings and yields clear gains in highly stochastic, non-stationary scenarios where optimization is more challenging. On relatively simple tasks, where MAPPO is already strong, the improvements from reflection are understandably less pronounced, but MARPO preserves comparable performance while offering stable learning dynamics. Future work includes exploring a broader range of hyperparameters, further refining the reflection mechanism, and extending MARPO to complex, large-scale tasks to fully exploit its advantages in challenging multi-agent environments.



\bibliography{aaai2026}

\end{document}